\documentclass[10pt]{amsart}

\usepackage{amssymb,amsthm,amsmath,enumerate}
\usepackage[numbers,sort&compress]{natbib}
\usepackage{color}
\usepackage{graphicx}
\usepackage{amssymb,amsthm,amsmath}
\usepackage[numbers,sort&compress]{natbib}
\usepackage{color}
\usepackage{graphicx}

\title[ ]{  Sharp   bounds for finitely many embedded  eigenvalues  of   perturbed Stark type operators   }

\author{Wencai Liu}
\address[Wencai Liu]{Department of Mathematics, University of California, Irvine, California 92697-3875, USA}\email{liuwencai1226@gmail.com}
\theoremstyle{plain}
\newtheorem{theorem}{Theorem}[section]

\newtheorem{lemma}[theorem]{Lemma}
\newtheorem{proposition}[theorem]{Proposition}

\newcommand{\R}{\mathbb{R}}
\newcommand{\Z}{\mathbb{Z}}
\theoremstyle{definition}

\newtheorem{remark}[theorem]{Remark}

\begin{document}


\begin{abstract}
For   perturbed Stark   operators
  $Hu=-u^{\prime\prime}-xu+qu$,
the author has proved  that $\limsup_{x\to \infty}{x}^{\frac{1}{2}}|q(x)|$  must be larger than $\frac{1}{\sqrt{2}}N^{\frac{1}{2}}$  in order to create $N$  linearly independent   eigensolutions  in  $L^2(\R^+)$  \cite{liustark}.
In this paper,  we apply {\it generalized  Wigner-von Neumann  type }functions to construct embedded eigenvalues for a class of Schr\"odinger operators, including   a proof that the bound $\frac{1}{\sqrt{2}}N^{\frac{1}{2}}$ is sharp.
\end{abstract}
\maketitle
\section{Introduction}
The  Stark  operator
$Hu=-u^{\prime\prime}-xu +qu$
   describes a charged quantum particle in a constant electric field with an additional electric potential $q$.
It has attracted  a lot of attentions from both mathematics and physics \cite{christ2003absolutely,kiselev2000absolutely,solem1997variations,epstein1926stark,courtney1995classical,py1,Kor1,Kor2,Graf1,Her2,Her1,Her5,Ya1,Jen1}.

In this paper, we consider a class of more general operators, Stark  type  operators on $L^2(\R^+)$:
\begin{equation}\label{Gpwstark}
   Hu=-u^{\prime\prime}-x^{\alpha}u+qu,
\end{equation}
where $0<\alpha<2$. Denote by $H_0u=-u^{\prime\prime}-x^{\alpha}u$ and regard $q$ as a perturbation.

It is well known that for any $0<\alpha<2$,  $\sigma _{\rm ess}(H_0)=\sigma _{\rm ac}(H_0)=\R$ and $H_0$ does not have any eigenvalue.
The criteria    for the perturbation  such  that  the associated perturbed Stark type operator has single  eigenvalue, finitely many   eigenvalues or countably  many   eigenvalues have been obtained in \cite{liustark}.

Define $P\subset \R$ as
 \begin{equation*}
    P=\{E\in\R: -u^{\prime\prime}-x^{\alpha}u+qu=Eu \text{ has  an }  L^2(\R^+) \text{ solution}\} .
 \end{equation*}
In   \cite{liustark}, the author proved that
\begin{theorem}\cite[Theorem 1.5]{liustark}\label{Maintheoremapr3apr29}
 Let $a$  be given by
\begin{equation}\label{oct271}
    a= \limsup_{x\to \infty}{x}^{1-\frac{\alpha}{2}}|q(x)|.
 \end{equation}
 Then we have
 \begin{equation}\label{Indexhalft}
a\geq \frac{2-\alpha}{\sqrt{2}}(\# P)^{\frac{1}{2}}.
 \end{equation}
\end{theorem}

\begin{theorem}\cite[Theorem 1.6]{liustark}\label{Mainthm3}
For any  $\{ E_j\}_{j=1}^N\subset \R$ and any $\{\theta_j\}_{j=1}^N\subset [0,\pi]$,
there exist  potentials   $q\in C^{\infty}[0,+\infty)$  such that
\begin{equation*}
 \limsup_{x\to \infty}   {x}^{1-\frac{\alpha}{2}} |q(x)|\leq (2-\alpha)e^{2\sqrt{\ln N}}N,
\end{equation*}
 and for any $j=1,2,\cdots,N$, $-u^{\prime\prime}-x^{\alpha}u+qu=E_ju$ has an $L^2(\R^+)$ solution  $u$ with the boundary condition
 \begin{equation*}
   \frac{u^{\prime}(0)}{u(0)}=\tan \theta_j.
\end{equation*}

\end{theorem}

Theorem \ref{Maintheoremapr3apr29} implies that in order to create $N$ linearly independent  eigensolutions in  $L^2(\R^+)$, the quantity $a$ given by \eqref{oct271} must be equal or  larger than $\frac{2-\alpha}{\sqrt{2}}N^{\frac{1}{2}}$.
However,  Theorem \ref{Mainthm3} shows that if we allow $a\geq(2-\alpha)e^{2\sqrt{\ln N}}N$, one can create   $N$ eigensolutions in  $L^2(\R^+)$  for arbitrary  $N$. There is a gap between $N^{\frac{1}{2}}$ and $ e^{2\sqrt{\ln N}}N$. It is natural to ask what is the sharp bound  of $a$ to create $N$ linearly independent eigensolutions in  $L^2(\R^+)$.

{\bf Question 1:}   What is the minimum of  $\gamma$ such that for any $N$,
  there is a potential $q$ on $\R^+$ such that $\# P\geq N$ and
\begin{equation*}
 \limsup_{x\to \infty}   {x}^{1-\frac{\alpha}{2}} |q(x)|\leq C(\gamma)N^{\gamma}.
\end{equation*}
Theorems \ref{Maintheoremapr3apr29} and \ref{Mainthm3} imply $\gamma\in[\frac{1}{2},1]$.

 Our first result in this paper is to show that for any $\alpha$ satisfying $\frac{2}{3}<\alpha<2$,  $\gamma=\frac{1}{2}$ is the  solution  to  Question 1. 
 \begin{theorem}\label{Mthm1}
 Suppose $\frac{2}{3}<\alpha<2$. Then
 for any $N>0$, there exists a potential $q$ on $\R^+$    such that
\begin{equation}\label{Defa}
     \limsup_{x\to \infty}{x}^{1-\frac{\alpha}{2}}|q(x)|\leq 48 (2-\alpha)\sqrt{N \ln N}
 \end{equation}
 and
  $\# P= N$.
\end{theorem}
For some technical reasons, currently we can only give the proof for $\frac{2}{3}<\alpha<2$. We believe it that  $\gamma=\frac{1}{2}$ is the  solution  to  Question 1  for all $0<\alpha<2$.

Question 1 and Theorem \ref{Mthm1} do not care about the locations of the corresponding energies.
If we  take the distribution of energies into consideration, what is the sharp upper bound? We formulate it as the following question.

{\bf Question 2:}   What is the minimum of  $\gamma$ such that for any $\{E_j\}_{j=1}^N$,
  there exists  a potential $q$ on $\R^+$ such that
$-u^{\prime\prime}-x^{\alpha}u+qu=E_ju$ has an $L^2(\R^+)$ solution  for  each $j=1,2,\cdots,N$ and
\begin{equation*}
 \limsup_{x\to \infty}   {x}^{1-\frac{\alpha}{2}} |q(x)|\leq C(\gamma)N^{\gamma}.
\end{equation*}
Theorems \ref{Maintheoremapr3apr29} and \ref{Mainthm3} imply $\gamma\in[\frac{1}{2},1]$.
We conjecture that $\gamma=1$ is the solution to Question 2.

During the  proof Theorem \ref{Mthm1}, we are able to improve the bound in Theorem \ref{Mainthm3}.
 \begin{theorem}\label{Mthm2}

For any $\varepsilon>0$,  $\{ E_j\}_{j=1}^N\subset \R$ and $\{\theta_j\}_{j=1}^N\subset [0,\pi]$,
there exist  local $L^1(\R^+)$ potentials  $q$   such that
\begin{equation}\label{Ggoalb}
 \limsup_{x\to \infty}   {x}^{1-\frac{\alpha}{2}} |q(x)|\leq (2-\alpha+\varepsilon) N,
\end{equation}
 and for each $j=1,2,\cdots,N$, $-u^{\prime\prime}-x^{\alpha}u+qu=E_ju$ has an $L^2(\R^+)$ solution  $u$ with the boundary condition
 \begin{equation*}
   \frac{u^{\prime}(0)}{u(0)}=\tan \theta_j.
\end{equation*}

\end{theorem}
\begin{remark}
\begin{itemize}
  \item Theorem \ref{Mthm2} gives better bounds than those in Theorem \ref{Maintheoremapr3apr29}, but less regularity in potentials.
  \item Applying  additional     piecewise constructions  in the proof of Theorem \ref{Mthm2}, it is possible to show that the upper bound in \eqref{Ggoalb} can be improved to $(2-\alpha)N$.
We refer the readers to the critical case of \cite[Theorem 1.2 ]{liustark}  for details.
\end{itemize}

\end{remark}

The proof of both Theorems \ref{Mthm1} and \ref{Mthm2} are inspired by the methods tackling  perturbed free  Schr\"odinger operators.
Let us  turn to   perturbed   free  Schr\"odinger operators $-D^2+V$ first.
 Naboko \cite{nabdense} and Simon  \cite{simdense} constructed   power-decaying  potentials   $V$  such that $-D^2+V$ has dense eigenvalues. Before that, Wigner-von Neumann  type functions can only create one $L^2$ solution \cite{von}. Recently, there have been several important  developments on the problem of embedded  eigenvalues  for Schr\"odinger operators, Laplacians on manifolds or other models \cite{jl,ld1,liustark,Luk14,Luk13,Naboko2018,Naboko2016,Naboko2015,KumuraI,jl3,nabokojdea18,lukjmaa16}.
  For  perturbed   Stark type operators, under the rational independence assumption of  set $\{E_j\}, $ Naboko and Pushnitskii \cite{naboko1} constructed  operators with given a set   $\{E_j\}$ as embedded eigenvalues. The author \cite{liustark} constructed   perturbed   Stark type operators with any given  $\{E_j\}$ as a set of eigenvalues  with the quantitative bound (see Theorem \ref{Mainthm3}).
  However, the potential can not be given explicitly.  One of the  motivations of this paper is to approach the problem in an explicit way.

  In   \cite{simdense}, Simon used   Wigner-von Neumann  type functions
 $V(x)=\frac{a}{1+x}\sum_{j} \sin(2\lambda_jx +2\phi_j)\chi_{[a_j,\infty)}$,
to complete his constructions. It turns out that   Wigner-von Neumann  type function is a   good way to create embedded eigenvalues \cite{Luk14,Luk13,Naboko2018,Naboko2016,Naboko2007,luktams15}.
 Moreover, Wigner-von Neumann  type functions can also be used to achieve  the optimal bounds.
 Denote by
 \begin{equation*}
   S=\{E>0: -u^{\prime\prime}(x)+V(x)u(x)=Eu(x)\text{ has an } L^2(\R^+) \text{ solution }\}.
 \end{equation*}
  Kiselev-Last-Simon \cite{KLS} proved if $\limsup_{x\to \infty}x|V(x)|<\infty$, then the set $S$ is countable and
 \begin{equation}\label{Gmay81}
   \sum_{E_i\in S}E_i<\infty.
 \end{equation}
 This result has been extended to perturbed periodic operators by the author \cite{liuasy}.
 By   Wigner-von Neumann  type functions and additional probability arguments  from \cite{prob}, Remling \cite{Remduke} proved that
 there are potentials $V(x)=O(x^{-1})$ with $ \sum_{E_i\in S}E_i^p=\infty$ for every $p<1$.
   Remling's result implies that \eqref{Gmay81} can not be improved in some sense, which  answers a question in \cite{KLS}.

  Another motivation of the present paper   is   to  find the substitution of  Wigner-von Neumann  type functions to deal with  perturbed Stark type operators, so that we can use the ideas of
   Simon,  Remling, and among others  to address our problems.
   We found   a good type of substituted functions $\frac{\sin(\int_0^{\xi} \sqrt{1-\beta_{E_j}(x)}dx+t_j)}{\xi}$, which is called the {\it generalized Wigner-von Neumann  type function}. See the definition of $\beta_{E_j}(x)$  below.

   Although the arguments in this paper are inspired by those in dealing with  perturbed free Schr\"odinger operators, the details are much more delicate and difficult, in particular,  oscillated integrals and  resonant phenomena. The spectra of perturbed free Schr\"odinger operators and Stark operators are quite different.
   Under the   assumption $V(x)=O(x^{-1})$, perturbed free Schr\"odinger operators can have infinitely many eigenvalues. However, under  the corresponding assumption, perturbed Stark operators can only have finitely many eigenvalues by Theorem \ref{Maintheoremapr3apr29}.  Moreover,
   for   free cases, there is only one leading entry  dominating  each  Pr\"ufer angle and the leading entries are  distinguished  by  energies.  For Stark type cases, there are finitely many  entries  (the number of entries  depends on $\alpha$) dominating  each  Pr\"ufer angle, and
  the first/leading entry is $1$ for any  Pr\"ufer angle, which   leads to resonance. Similar resonance has been  studied in    \cite{liustark}.
  In the proof of  Theorem \ref{Mthm2}, we are able to deal with the resonance coming from the first dominating entry and all the other dominating  entries at the same time. However, for the technical reason,   we can only deal with  the first two dominating entry  for the topic in Theorem \ref{Mthm1}, and the assumption $\alpha>\frac{2}{3}$ will guarantee that  there are  exact  two entries to dominate Pr\"ufer angles.

\section{Preparations}
Let $v_{\alpha}(x)=x^{\alpha} $  for  $x\in\R^+$   and consider the Schr\"odinger equation on $\R^+$,
 \begin{equation}\label{Gsch}
 -u^{\prime\prime}(x)-x^{\alpha}u(x)+q(x)u(x)=Eu(x).
 \end{equation}
The  Liouville transformation (see \cite{christ2003absolutely,naboko1}) is given by
\begin{equation}\label{GLiou}
  \xi(x)=\int_0^x\sqrt{v_{\alpha}(t)} dt, \phi(\xi)=v_{\alpha}(x(\xi))^{\frac{1}{4}}u(x(\xi)).
\end{equation}
We  define a weight function $p_{\alpha}(\xi)$ by
\begin{equation}\label{GWei}
  p_{\alpha}(\xi)= \frac{1}{v_{\alpha}(x(\xi))}.
\end{equation}
We also define a potential   by
\begin{equation}\label{GPQ}
 Q_{\alpha}(\xi,E)= -\frac{5}{16}\frac{|v_{\alpha}^{\prime}(x(\xi))|^2}{v_{\alpha}(x(\xi))^3}+\frac{1}{4}\frac{v_{\alpha}^{\prime\prime}(x(\xi))}{v_{\alpha}(x(\xi))^2} +\frac{q(x(\xi))-E}{v_{\alpha}(x(\xi))}.
\end{equation}
 Let $c=(1+\frac{\alpha}{2})^{\frac{2}{2+\alpha}}$.
 Direct computations imply that
 \begin{equation}\label{GLiou1}
  x=c\xi^{\frac{2}{2+\alpha}}, \phi(\xi,E)=c^{\frac{\alpha}{4}}\xi^{\frac{\alpha}{2(2+\alpha)}}u(c\xi^{\frac{2}{2+\alpha}}),
 \end{equation}
 \begin{equation}\label{GWei1}
  p(\xi)= \frac{1}{c^{\alpha}\xi^{\frac{2\alpha}{2+\alpha}}},
\end{equation}
and
\begin{equation}\label{GPQ1old}
 Q_{\alpha}(\xi,E)= -\frac{5}{4}\frac{\alpha^2}{(2+\alpha)^2}\frac{1}{\xi^2}+\frac{\alpha(\alpha-1)}{(2+\alpha)^2}\frac{1}{\xi^{2}}+\frac{q(c\xi^{\frac{2}{2+\alpha}})-E}{c^{\alpha}\xi^{\frac{2\alpha}{2+\alpha}}}.
\end{equation}

Notice that  the potential $Q_{\alpha}(\xi,E)$ depends on $q$, $\alpha$ and $E$.
In the following, we always fix $\alpha\in(0,2)$. For simplicity, we drop off its dependence.
Let
\begin{equation}\label{Gvapr}
 V(\xi)=\frac{q(c\xi^{\frac{2}{2+\alpha}})}{c^{\alpha}\xi^{\frac{2\alpha}{2+\alpha}}}.
\end{equation}
Then
\begin{eqnarray}
   Q(\xi,E) &=& -\frac{5}{4}\frac{\alpha^2}{(2+\alpha)^2}\frac{1}{\xi^2}+\frac{\alpha(\alpha-1)}{(2+\alpha)^2}\frac{1}{\xi^{2}}-\frac{E}{c^{\alpha}\xi^{\frac{2\alpha}{2+\alpha}}}+V(\xi)\label{GPQ1} \\
   &=& -\frac{E}{c^{\alpha}\xi^{\frac{2\alpha}{2+\alpha}}}+V(\xi)+\frac{O(1)}{\xi^{2}}.\nonumber
\end{eqnarray}


Suppose $u\in L^2(\R^+)$  is a solution of \eqref{Gsch}. It follows that
$\phi$ satisfies
\begin{equation}\label{Gschxi}
  -\frac{d^2\phi}{d\xi^2}+Q(\xi,E)\phi=\phi,
\end{equation}
and $\phi\in L^2(\R ^+,p(\xi)d\xi)$.

Below,
 $\epsilon>0$   always  depends on $\alpha$  in an explicit way. Denote by
 \begin{equation*}
  \beta_E(\xi)=  -\frac{E}{c^{\alpha}\xi^{\frac{2\alpha}{2+\alpha}}}.
\end{equation*}
When $\xi$ is large, one has $|\beta_E(\xi)|<1$.
By shifting  the equation, we always assume $\beta_E(\xi)$ is sufficiently small. So $\sqrt{1-\beta_E(\xi)}$ is well defined.
\begin{proposition}\label{Keyprop}
Suppose $a\neq0$.
Then following estimates hold  for $\xi>\xi_0>1$ and $\gamma\in \R$:
\begin{enumerate}
  \item
  \begin{equation}\label{Gosc1}
  \int_{\xi_0}^{\xi} \frac{\sin (a\int_0^{s}\sqrt{1-\beta_E(x)}dx+\gamma)}{s}ds=\frac{O(1)}{\xi_0^{\epsilon}}.
\end{equation}
  \item for any  $E_1\neq E_2\in\R$,
  \begin{equation}\label{Gosc2}
  \int_{\xi_0}^{\xi}\frac{\sin (a\int_0^{s}\sqrt{1-\beta_{E_1}(x)}dx\pm a\int_0^{s}\sqrt{1-\beta_{E_2}(x)}dx+\gamma)}{s}ds=\frac{O(1)}{\xi_0^{\epsilon}}.
  \end{equation}
\end{enumerate}
\end{proposition}
\begin{proof}
We only give the proof of  case ``$-$" in \eqref{Gosc2}.  The rest  can be proceeded  in a similar way.

Denote by $\beta(\xi)=a\int_0^{\xi}[\sqrt{1-\beta_{E_1}(x)}dx- a\int_0^{\xi}\sqrt{1-\beta_{E_2}(x)}]dx+\gamma$. Then
\begin{equation*}
  \beta^{\prime}(\xi)= a\sqrt{1-\beta_{E_1}(\xi)}-a\sqrt{1-\beta_{E_2}(\xi)}=a\frac{E_1-E_2}{2c^{\alpha}\xi^{\frac{2\alpha}{2+\alpha}}}+\frac{o(1)}{\xi^{\frac{2\alpha}{2+\alpha}}},
\end{equation*}
and
\begin{equation*}
  \beta^{\prime\prime}(\xi)=\frac{O(1)}{\xi^{1+\frac{2\alpha}{2+\alpha}}}.
\end{equation*}
Integration by part, we have
$$\int_{\xi_0}^{\xi}\frac{\sin (a\int_0^{s}\sqrt{1-\beta_{E_1}(x)}dx- a\int_0^{s}\sqrt{1-\beta_{E_2}(x)}dx+\gamma)}{s}$$
\begin{eqnarray}
   &=& \int_{\xi_0}^{\xi}\frac{\sin\beta(s)}{s}ds=\int_{\xi_0}^{\xi}\frac{\beta^{\prime}(s)\sin\beta(s)}{\beta^{\prime}(s)s}ds \nonumber\\
  &=&   O(\xi_0^{-\epsilon})+ O(1)\int_{\xi_0}^{\xi} \frac{\cos\beta(s)\beta^{\prime\prime}}{s\beta^{\prime 2}}ds \nonumber\\
   &=&    O(\xi_0^{-\epsilon}).\nonumber
\end{eqnarray}
\end{proof}
\section{Asymptotical behavior  of  solutions for a  class of linear systems}\label{Secasy}

The proof of this section is inspired by the WKB method. We refer the readers to papers \cite{dollard1977product,harris1975asymptotic,Levinson} for arguments.


\begin{theorem}\label{Keythm1}
Suppose $a>0$ is a constant. Suppose  $\{E_j\}\in \R$ are distinct.
Define $V(\xi)=0$ for $\xi\in[0,1]$ and
\begin{equation}\label{Gdef.V}
 V(\xi)=\frac{4 a}{\xi}\sum_{j=1}^N\sin\left(\int_0^{\xi}2\sqrt{1-\beta_{E_j}(x)}dx+2t_j\right),
\end{equation}
for $\xi>1$.

Define  $q(x)$ on $[0,\infty)$ such that \eqref{Gvapr} holds for \eqref{Gdef.V}.
Let $Q(\xi,E)$ be given by \eqref{GPQ1}.
 Then the following asymptotics hold as $\xi$ goes to infinity,
\begin{enumerate}
   \item
    if $E\neq  E_j$ for any $j=1,2,\cdots N$, then there exists a fundamental system of solutions $\{y_1(\xi),y_2(\xi)\}$ of \eqref{Gschxi} such that
    \begin{equation*}
    \left[
      \begin{array}{c}
         y_1(\xi) \\
         y_1^\prime(\xi)
       \end{array}
       \right]=
        \left[
      \begin{array}{c}
         \cos( \int_0^{\xi}\sqrt{1-\beta_E(x)}dx+t_j)\\
         - \sin( \int_0^{\xi}\sqrt{1-\beta_E(x)}dx+t_j)
       \end{array}
       \right]+O( \xi^{-\epsilon})
    \end{equation*}
    and
    \begin{equation*}
    \left[
      \begin{array}{c}
         y_2(\xi) \\
         y_2^\prime(\xi)
       \end{array}
       \right]=
        \left[
      \begin{array}{c}
         \sin( \int_0^{\xi}\sqrt{1-\beta_E(x)}dx+t_j) \\
         \cos(\int_0^{\xi}\sqrt{1-\beta_E(x)}dx+t_j)
       \end{array}
       \right]+O( \xi^{-\epsilon}).
    \end{equation*}

    \item
     if $E= E_j$ for some $j$, then there exists a fundamental system of solutions $\{y_1(\xi),y_2(\xi)\}$ of \eqref{Gschxi} such that
    \begin{equation*}
    \left[
      \begin{array}{c}
         y_1(\xi) \\
         y_1^\prime(\xi)
       \end{array}
       \right]=
        \xi^a\left[
      \begin{array}{c}
         \cos( \int_0^{\xi}\sqrt{1-\beta_{E_j}(x)}dx+t_j)\\
         - \sin(\int_0^{\xi} \sqrt{1-\beta_{E_j}(x)}dx+t_j)
       \end{array}
       \right]+O( \xi^{a-\epsilon})
    \end{equation*}
    and
    \begin{equation*}
    \left[
      \begin{array}{c}
         y_2(\xi) \\
         y_2^\prime(\xi)
       \end{array}
       \right]=
        \xi^{-a}\left[
      \begin{array}{c}
         \sin(\int_0^{\xi} \sqrt{1-\beta_{E_j}(x)}dx+t_j) \\
         \cos(\int_0^{\xi}\sqrt{1-\beta_{E_j}(x)}dx+t_j)
       \end{array}
       \right]+O( \xi^{-a-\epsilon}).
    \end{equation*}

\end{enumerate}

\end{theorem}
\begin{proof}
In order to avoid repetition, we only give the proof of the case
  the case $E= E_j$ for some $j=1,2,\cdots,N$.
Denote by
\begin{equation}\label{Gdef.Vmay2}
 \tilde{V}(\xi)=\frac{4 a}{\xi}\sum_{i=1,i\neq j}^N\sin(2\int_0^{\xi}\sqrt{1-\beta_{E_i}(x)}dx+2t_i).
\end{equation}
  Rewrite the second order differential equation of \eqref{Gschxi} as the linearly differential  equations,
  \begin{equation*}
    \left[\begin{array}{c}
      y \\
      y^{\prime}
    \end{array}
    \right]^\prime=\left[
                     \begin{array}{cc}
                       0 & 1 \\
                     \beta_{E_j}(\xi)+V(\xi)+\frac{O(1)}{\xi^{2}}-1 & 0 \\
                     \end{array}
                   \right] \left[\begin{array}{c}
      y \\
      y^{\prime}
    \end{array}
    \right].
  \end{equation*}
   Let \begin{equation*}
    \left[\begin{array}{c}
      u_1 \\
      u_2
    \end{array}
    \right]=\left[
                     \begin{array}{cc}
                      \sqrt{1-\beta_{E_j}(\xi)} &  0 \\
                      0 & 1 \\
                     \end{array}
                   \right] \left[\begin{array}{c}
      y \\
      y^{\prime}
    \end{array}
    \right].
  \end{equation*}
  We obtain  a new equation
\begin{equation*}
    \left[\begin{array}{c}
      u_1 \\
      u_2
    \end{array}
    \right]^{\prime}=\left(\left[
                     \begin{array}{cc}
                     0 &   \sqrt{1-\beta_{E_j}(\xi)} \\
                      -\sqrt{1-\beta_{E_j}(\xi)}+\frac{V}{ \sqrt{1-\beta_{E_j}(\xi)}}& 1 \\
                     \end{array}
                   \right]+\frac{O(1)}{\xi^{1+\epsilon}} \right)\left[\begin{array}{c}
      y \\
      y^{\prime}
    \end{array}
    \right].
  \end{equation*}
  Let \begin{equation*}
    \left[\begin{array}{c}
      y_1 \\
      y_2
    \end{array}
    \right]=\left[
                     \begin{array}{cc}
                       \cos(\int_0^{\xi}\sqrt{1-\beta_{E_j}(x)}dx+t_j) &  -\sin(\int_0^{\xi}\sqrt{1-\beta_{E_j}(x)}dx+t_j) \\
                       \sin(\int_0^{\xi}\sqrt{1-\beta_{E_j}(x)}dx+t_j) & \cos(\int_0^{\xi}\sqrt{1-\beta_{E_j}(x)}dx+t_j) \\
                     \end{array}
                   \right] \left[\begin{array}{c}
      u_1 \\
      u_2
    \end{array}
    \right].
  \end{equation*}
  Obviously,
  one has \begin{equation*}
    \left[\begin{array}{c}
      u_1 \\
     u_2
    \end{array}
    \right]=\left[
                     \begin{array}{cc}
                       \cos(\int_0^{\xi}\sqrt{1-\beta_{E_j}(x)}dx+t_j) &  \sin(\int_0^{\xi}\sqrt{1-\beta_{E_j}(x)}dx+t_j) \\
                       -\sin(\int_0^{\xi}\sqrt{1-\beta_{E_j}(x)}dx+t_j) & \cos(\int_0^{\xi}\sqrt{1-\beta_{E_j}(x)}dx+t_j) \\
                     \end{array}
                   \right] \left[\begin{array}{c}
      y_1 \\
      y_2
    \end{array}
    \right].
  \end{equation*}
   After some calculations,  we have
  \begin{eqnarray*}
    \left[\begin{array}{c}
      y_1 \\
      y_2
    \end{array}
    \right]^\prime 
     &= &  (\Lambda(\xi)+H(\xi)+O(\xi^{-1-\epsilon}))\left[\begin{array}{c}
      y_1 \\
      y_2
    \end{array}
    \right],
  \end{eqnarray*}
 where
 \begin{equation}\label{Glam}
   \Lambda(\xi)=\left[
                \begin{array}{cc}
                  -2a\frac{\sin^2(2\int_0^{\xi}\sqrt{1-\beta_{E_j}(x)}dx+2t_j)}{\xi} & 0 \\
                  0 & 2a\frac{\sin^2(2\int_0^{\xi}\sqrt{1-\beta_{E_j}(x)}dx+2t_j)}{\xi} \\
                \end{array}
              \right]
 \end{equation}
 and
 \begin{equation*}
    H(\xi)=\left[
                     \begin{array}{cc}
                       H_{11}(\xi) &  H_{12}(\xi) \\
                     H_{21}(\xi)   &   H_{22}(\xi)  \\
                     \end{array}
                   \right].
 \end{equation*}
The explicit formulas for $H_{ij}$, $i,j=1,2$ are
\begin{eqnarray*}
 H_{11}  &=& -\frac{1}{2} \tilde{V}(\xi)\sin \left(2\int_0^{\xi}\sqrt{1-\beta_{E_j}(x)}dx+2t_j\right)  ,
\end{eqnarray*}

\begin{equation*}
  H_{22}(\xi)= \frac{1}{2} \tilde{V}(\xi)\sin \left(2\int_0^{\xi}\sqrt{1-\beta_{E_j}(x)}dx+2t_j\right) ,
\end{equation*}
\begin{equation*}
  H_{12}(\xi)=-\frac{1}{2}V(\xi)\left(1-\cos \left(2\int_0^{\xi}\sqrt{1-\beta_{E_j}(x)}dx+2t_j\right)\right),
\end{equation*}
and
\begin{equation*}
  H_{21}(\xi)=\frac{1}{2}V(\xi)\left(1+\cos \left(2\int_0^{\xi}\sqrt{1-\beta_{E_j}(x)}dx+2t_j\right)\right).
\end{equation*}
By Proposition \ref{Keyprop}, one has
\begin{equation*}
  Q(\xi)\equiv-\int_{\xi}^{\infty}H(s)ds= O(\xi^{-\epsilon}).
\end{equation*}
  Assume $  \xi$ is large, then $||Q||\leq \frac{1}{2}$. Let $ \left[\begin{array}{c}
       \tilde{y}_1 \\
       \tilde{y}_2
    \end{array}
    \right]= (I+Q)^{-1}\left[\begin{array}{c}
      y_1 \\
      y_2
    \end{array}
    \right]  $.
    We  obtain
    \begin{equation}
        \left[\begin{array}{c}
       \tilde{y}_1 \\
      \tilde{y}_2
    \end{array}
    \right]^\prime=(\Lambda(\xi)+R(\xi))\left[\begin{array}{c}
      \tilde{y}_1 \\
      \tilde{y}_2
    \end{array}
    \right],\label{hatdiiequ}
    \end{equation}
    where $R(\xi)=(R_{ij})=O(\xi^{-1-\epsilon})$.
    Let $\varphi(\xi)= \left[\begin{array}{c}
       \tilde{y}_1 \\
      \tilde{y}_2
    \end{array}
    \right]$
    and
    $$\lambda(\xi)=2a\frac{\sin^2(2\int_0^{\xi}\sqrt{1-\beta_{E_j}(x)}dx+2t_j)}{\xi}.$$
  Let us consider the integral equation,
\begin{equation}\label{Ginte}
      \varphi(\xi)=\left[\begin{array}{c}
      e^{-\int_{1}^{\xi}\lambda(s)ds} \\
     0
    \end{array}
    \right]-\int_{\xi}^{\infty}\left[
                                 \begin{array}{cc}
                                  e^{-\int_{\xi}^{y}\lambda(s)ds} & 0 \\
                                   0 &e^{\int_{\xi}^{y}\lambda(s)ds} \\
                                 \end{array}
                               \right]R(y)\varphi(y)dy.
 \end{equation}
 If \eqref{Ginte} has a solution $||\varphi(\xi)||\leq 2e^{-\int_{1}^{\xi}\lambda(s)ds}$, then (by direct computation)
 $\varphi(\xi)$ is a solution of  equation \eqref{hatdiiequ}.
 Moreover

 \begin{eqnarray}
  || \int_{\xi}^{\infty}\left[
                                 \begin{array}{cc}
                                  e^{-\int_{\xi}^{y}\lambda(s)ds} & 0 \\
                                   0 &e^{\int_{\xi}^{y}\lambda(s)ds} \\
                                 \end{array}
                               \right]R(y)\varphi(y)dy ||&=& O(1) \int_{\xi}^{\infty} e^{\int_{\xi}^y\lambda(s)}||R(y)||e^{-\int_{1}^y\lambda(s)ds} \nonumber\\
    &=& O(1) e^{-\int_{1}^{\xi}\lambda(s)ds}\int_{\xi}^{\infty} ||R(y)||\nonumber\\
    &=&  e^{-\int_{1}^{\xi}\lambda(s)ds} O(\xi^{-\epsilon}). \label{Gmay21}
 \end{eqnarray}
  Define  iteration equations:
\begin{equation}\label{Ginteir}
      \varphi_k(\xi)=\left[\begin{array}{c}
      e^{-\int_{1}^{\xi}\lambda(s)ds} \\
     0
    \end{array}
    \right]-\int_{\xi}^{\infty}\left[
                                 \begin{array}{cc}
                                  e^{-\int_{\xi}^{y}\lambda(s)ds} & 0 \\
                                   0 &e^{\int_{\xi}^{y}\lambda(s)ds} \\
                                 \end{array}
                               \right]R(y)\varphi_{k-1}(y)dy,
 \end{equation}
 with $\varphi=0$.
 By induction that $||\varphi_k(\xi)-\varphi_{k-1}(\xi)||\leq \frac{1}{2^{k+1}}e^{-\int_{1}^{\xi}\lambda(s)ds}$, one can show \eqref{Ginte} has a   solution $||\varphi(\xi)||\leq 2e^{-\int_{1}^{\xi}\lambda(s)ds}$ (see p.94 in
  \cite{Levinson} for all the details).

  By \eqref{hatdiiequ}, \eqref{Ginte} and \eqref{Gmay21}, we get a solution
  \begin{equation*}
    \left[\begin{array}{c}
       \tilde{y}_1 \\
      \tilde{y}_2
    \end{array}
    \right]=e^{-\int_1^{\xi}\lambda(s)ds}\left(\left[\begin{array}{c}
       1 \\
     0
    \end{array}
    \right]+O(\xi^{-\epsilon})\right).
  \end{equation*}
By Proposition \ref{Keyprop} again, we have
\begin{eqnarray*}
  \int_1^{\xi}\lambda(s)ds&=&  2a\int_1^{\xi}\frac{\sin^2(2\int_0^{s}\sqrt{1-\beta_{E_j}(x)}dx+2t_j)}{s}ds\\
 &=& a\int_1^{\xi}\frac{1}{s}ds -a\int_1^{\xi}\frac{\cos(4\int_0^{s}\sqrt{1-\beta_{E_j}(x)}dx+4t_j)}{s}ds\\
 &=& a\int_1^{\xi}\frac{1}{s}ds -a\int_1^{\infty}\frac{\cos(4\int_0^{s}\sqrt{1-\beta_{E_j}(x)}dx+4t_j)}{s}ds\\
 &&+a\int_{\xi}^{\infty}\frac{\cos(4\int_0^{s}\sqrt{1-\beta_{E_j}(x)}dx+4t_j)}{s}ds\\
  &=&\ln \xi^a-c+O(\xi^{-\epsilon}),
\end{eqnarray*}
where the constant $c$ equals
$$ a\int_1^{\infty}\frac{\cos(4\int_0^{s}\sqrt{1-\beta_{E_j}(x)}dx+4t_j)}{s}ds.$$
Thus \eqref{hatdiiequ} has a solution
\begin{equation}\label{Gmay22}
    \left[\begin{array}{c}
       \tilde{y}_1 \\
      \tilde{y}_2
    \end{array}
    \right]=\xi^{-a}\left[\begin{array}{c}
       1 \\
     0
    \end{array}
    \right]+O(\xi^{-a-\epsilon}).
  \end{equation}
  By the similar argument (see
  \cite{Levinson} again), we obtain that
  \eqref{hatdiiequ} has a solution
\begin{equation}\label{Gmay23}
    \left[\begin{array}{c}
       \tilde{y}_1 \\
      \tilde{y}_2
    \end{array}
    \right]=\xi^{a}\left[\begin{array}{c}
       1 \\
     0
    \end{array}
    \right]+O(\xi^{a-\epsilon}).
  \end{equation}
Now the Theorem follows from  \eqref{Gmay22} and \eqref{Gmay23}.
  \end{proof}
\section{Proof of Theorems \ref{Mthm1} and \ref{Mthm2}}

\begin{proof}[\bf Proof of Theorem \ref{Mthm2}]
Let   $a>\frac{2-\alpha}{2(2+\alpha)}$.
Define potentials
\begin{equation*}
 V(\xi)=\frac{4 a}{\xi}\sum_{j=1}^N\sin(2\int_0^{\xi}\sqrt{1-\beta_{E_j}(x)}dx+2t_j)\chi_{[a_j,\infty)},
\end{equation*}
with large $a_j$.

Obviously,
\begin{equation*}
   |\xi V(\xi)|\leq 4aN,
\end{equation*}
so that
\begin{equation}\label{Gmay34}
   |\xi^{1-\frac{\alpha}{2}} q(\xi)|\leq 2(2+\alpha)aN.
\end{equation}
By \eqref{GWei1},
\begin{equation}\label{Gmay33aug}
     p(\xi)|\phi(\xi,E_j)|^2\leq O(1)\xi^{-2a-\frac{2\alpha}{2+\alpha}}\leq O(1)\xi^{-1-\epsilon}.
\end{equation}
By   \eqref{Gmay33aug}, \eqref{Gschxi} has a solution $\phi(\xi,E_j)\in L^2(\R ^+,p(\xi)d\xi)$ for each $j=1,2,\cdots,N$.
However, $\phi(\xi,E_j)$ may not satisfy the given boundary condition. This can be done by adjusting $a_j$ and additional functions $W$ with support in $(1,2)$.
We refer   readers to \cite{simdense} for  rigorous arguments.
Now the Theorem follows from \eqref{Gmay34}.
\end{proof}
Let
\begin{equation*}
  \tau= \frac{2+\alpha}{2(2-\alpha)}\frac{1}{(1+\frac{\alpha}{2})^{\frac{2\alpha}{2+\alpha}}},
\end{equation*}
so that
\begin{equation*}
\tau \frac{d\xi^{\frac{2-\alpha}{2+\alpha}}}{d\xi}=\frac{1}{2}\frac{1}{c^{\alpha}\xi^{\frac{2\alpha}{2+\alpha}}}.
\end{equation*}
For any given $N$,
let
\begin{equation*}
  E_j=\frac{j}{N\tau}, \text{ for } j=1,2,\cdots,N.
\end{equation*}
Before giving the proof of Theorem \ref{Mthm1}, one key Lemma is needed, which is motivated  by  \cite{Remduke}.
\begin{lemma}\label{Lemapr30}
 There  exist $\theta_j\in[0,1)$ for $j=1,2,\cdots,N$ such that  for any $\xi>0$
\begin{equation}\label{Gmax}
  |\sum_{j=1}^N\sin(2\tau E_j\xi+2\pi\theta_j)|+|\sum_{j=1}^N\cos(2\tau E_j\xi+2\pi \theta_j)|\leq4\sqrt{2N\ln (8(N+1)N)}.
\end{equation}
\end{lemma}
\begin{proof}
Let
\begin{equation}\label{GdefE}
  \tilde{E}_j=\frac{j}{N}, \text{ for } j=1,2,\cdots,N.
\end{equation}
It suffices to show
there  exist $\theta=(\theta_1,\theta_2,\cdots,\theta_N)\in[0,1)^N$ such that  for any $\xi>0$,
\begin{equation}\label{Gmax1}
f(\xi,\theta) \equiv |\sum_{j=1}^N \sin(2 \tilde{E}_j\xi+2\pi\theta_j)|+|\sum_{j=1}^N\cos(2  \tilde{E}_j\xi+2\pi\theta_j)| \leq4\sqrt{2N\ln (8(N+1)N)}.
\end{equation}
Let
\begin{equation*}
  f_1(\xi,\theta)=\sum_{j=1}^N \sin(2 \tilde{E}_j\xi+2\pi\theta_j),
\end{equation*}
and
\begin{equation*}
  M_1(\theta)=\sup_{\xi\in \R^+}|f_1(\xi,\theta)|=\max_{0\leq \xi \leq \pi N}|f_1(\xi,\theta)|.
\end{equation*}
Pick $\xi_0$ so that $|f_1(\xi_0,\theta)|= M_1(\theta)$. Notice that
\begin{equation*}
  1=\frac{2}{\pi N}\int_{0}^{\pi N}f_1(\xi,\theta)\sin(2 \tilde{E}_j\xi+2\pi\theta_j)d\xi\leq 2M_1(\theta),
\end{equation*}
and
\begin{eqnarray*}
  \left|\frac{d f_{1}(\xi,\theta)}{d\xi}\right| &=& 2|\sum_{j=1}^N \tilde{E}_j\cos(2 \tilde{E}_j\xi+\theta_j) | \\
   &\leq& 2|\sum_{j=1}^N \tilde{E}_j| \\
   &=&  N+1.
\end{eqnarray*}
It follows that there exits a interval $I(\theta) $ of $\xi$ centered at $\xi_0$  with size $ \frac{1}{2 (N+1)}$ ($I(\theta)$ means the interval depends on $\theta$) such that
\begin{equation*}
 | f_1(\xi,\theta)|\geq M_1(
 \theta)-\frac{1}{4(N+1)}(N+1)\geq \frac{M_1(\theta)}{2},
\end{equation*}
for all $\xi\in I(\theta)$.

Integration yields  (p.493, \cite{Remduke})
 \begin{equation*}
    \int_0^{2\pi}e^{a\sin(b+2\pi y)}dy\leq e^{\frac{a^2}{4}}.
 \end{equation*}
Thus
\begin{eqnarray*}
 \frac{1}{2(N+1)} \int_{(\R/  \Z)^N}e^{\frac{1}{2}\lambda M_1(\theta)}d\theta &\leq &   \int_{(\R/ \Z)^N}  d\theta\int_{I(\theta)}(e^{ \lambda f_1(\xi,\theta)} +e^{- \lambda f_1(\xi,\theta)})d\xi\\
   &\leq&   \int_{(\R/\Z)^N}  d\theta\int_{0}^{\pi N}(e^{ \lambda f_1(\xi,\theta)} +e^{- \lambda f_1(\xi,\theta)})d\xi \\
   &= &  \int_{0}^{\pi N} d\xi\int_{(\R/\Z)^N} (e^{ \lambda f_1(\xi,\theta)} +e^{- \lambda f_1(\xi,\theta)}) d\theta\\
    &\leq &  \int_{0}^{\pi N}  2 e^{\frac{N}{4}\lambda^2}d\xi\\
     &= &     2\pi N e^{\frac{N}{4}\lambda^2}.
\end{eqnarray*}
It implies
\begin{equation}\label{Gaug221}
   \int_{(\R/\Z)^N}e^{\frac{1}{2}\lambda (M_1(\theta)-\frac{N}{2}\lambda-\frac{2}{\lambda}\ln 8N(N+1)-\frac{2}{\lambda}\ln2\pi)}d\theta \leq \frac{1}{4}.
\end{equation}
By \eqref{Gaug221}, we have for any $\lambda>0$, there exists a subset $S_1(\lambda)\subset [0,1)^{N}$ such that ${\rm Leb}(S_1)\geq \frac{3}{4}$ and
\begin{equation*}
  M_1(\theta)\leq \frac{N}{2}\lambda+\frac{2}{\lambda}\ln 8N(N+1))+\frac{2}{\lambda}\ln2\pi,
\end{equation*}
for all $\theta\in S_1$.

Let
\begin{equation*}
  f_2(\xi,\theta)=\sum_{j=1}^N \cos(2 \tilde{E}_j\xi+2\pi\theta_j),
\end{equation*}
and
\begin{equation*}
  M_2(\theta)=\sup_{\xi\in \R^+}|f_1(\xi,\theta)|=\max_{0\leq \xi \leq \pi N}|f_2(\xi,\theta)|.
\end{equation*}
Similarly,  for any $\lambda>0$, there exists a subset $S_2(\lambda)\subset [0,1)^{N}$ such that ${\rm Leb}(S_2)\geq \frac{3}{4}$ and
\begin{eqnarray*}
  M_2(\theta) &\leq&\frac{N}{2}\lambda+\frac{2}{\lambda}\ln 8N(N+1))+\frac{2}{\lambda}\ln 2\pi \\
   &\leq& M_2(\theta)\leq\frac{N}{2}\lambda+\frac{4}{\lambda}\ln 8N(N+1)).
\end{eqnarray*}
for all $\theta\in S_2$.

Let
\begin{equation*}
  \lambda=\sqrt{\frac{8\ln 8N(N+1)}{N}},
\end{equation*}
and  $S=S_1(\lambda)\cap S_2(\lambda)$.
Then we have ${\rm Leb}(S)\geq \frac{1}{4}$, and
\begin{equation*}
 M_1(\theta)+M_2(\theta)\leq  4\sqrt{2N\ln 8N(N+1)},
\end{equation*}
for all $\theta\in S.$ We finish the proof.
\end{proof}
\begin{proof}[\bf Proof of Theorem \ref{Mthm1}]
It suffices to prove the case that  $N\geq 2$.
For any given $N\geq 2$,
let
\begin{equation*}
  E_j=\frac{j}{N\tau}, \text{ for } j=1,2,\cdots,N.
\end{equation*}
Let $a>\frac{2-\alpha}{2(2+\alpha)}$.
By Lemma \ref{Lemapr30},
 there exist  $t_j\in[0,\pi)$ such that  for any $\xi>0$,
\begin{equation}\label{Gmaxmay}
  |\sum_{j=1}^N\sin(2\tau E_j\xi+2t_j)|+|\sum_{j=1}^N\cos(2\tau E_j\xi+2t_j)|\leq 4\sqrt{2N\ln (8(N+1)N)}.
\end{equation}
Let
\begin{equation}\label{Gdef.Vmay}
 V(\xi)=\frac{4 a}{\xi}\sum_{j=1}^N\sin\left(2\int_{0}^{\xi}\sqrt{1-\beta_{E_j}(x)}dx+2t_j\right).
\end{equation}
By Tayor series, one has
\begin{eqnarray}
  \int_{0}^{\xi}\sqrt{1-\beta_{E_j}(x)}dx &=& \int_{0}^{\xi}\left( 1-\frac{1}{2}\beta_{E_j}(x)+O(1)\beta_{E_j}(x)^2\right)dx \nonumber\\
  &=& \xi+\tau E_j\xi^{\frac{2-\alpha}{2+\alpha}}+\frac{O(1)}{\xi^{\frac{4\alpha}{2+\alpha}-1}}+\tilde{t}_j\nonumber\\
   &=& \xi+\tau E_j\xi^{\frac{2-\alpha}{2+\alpha}}+O(\xi^{-\epsilon})+\tilde{t}_j,\label{Gaug22}
\end{eqnarray}
 since $\alpha>\frac{2}{3}$.

By \eqref{Gmaxmay}, \eqref{Gdef.Vmay} and \eqref{Gaug22},
 one has
 \begin{eqnarray}
   |\xi V(\xi)| &=& 4a|\sum_{j=1}^N\sin\left(2\int_{0}^{\xi}\sqrt{1-\beta_{E_j}(x)}dx+2t_j\right)| \nonumber\\
   &=&  4a|\sum_{j=1}^N\sin\left(2\xi+2\tau E_j\xi^{\frac{\alpha}{2+\alpha}}+2t_j+2\tilde{t}_j\right)|+O(\xi^{-\epsilon}) \nonumber\\
    &\leq& 4a|\sin(2\tau E_j\xi^{\frac{\alpha}{2+\alpha}}+2t_j+2\tilde{t}_j)|+4a|\cos(2\tau E_j\xi^{\frac{\alpha}{2+\alpha}}+2t_j+2\tilde{t}_j)|+O(\xi^{-\epsilon})\nonumber\\
     &\leq& 16a\sqrt{2N\ln (8(N+1)N)}+O(\xi^{-\epsilon})\nonumber\\
       &\leq& 96a\sqrt{N\ln N},\label{Gdef.Vmay21}
 \end{eqnarray}
 for large $\xi$.
Define  $q(x)$ on $[0,\infty)$ such that \eqref{Gvapr} holds for \eqref{Gdef.Vmay}.
Then by \eqref{Gdef.Vmay21},
 we have
 \begin{equation}\label{Gmay31}
   \xi^{1-\frac{\alpha}{2}}|q(\xi)|\leq  (2+\alpha)48a\sqrt{N\ln N}.
 \end{equation}
By Theorem \ref{Keythm1}, for any $E_j$, $j=1,2,\cdots,N$, \eqref{Gschxi} has a solution $\phi(\xi,E_j)$ satisfying
\begin{equation*}
    |\phi(\xi,E_j)|\leq 2\xi^{-a}
\end{equation*}
for large $\xi$.
By \eqref{GWei1},
\begin{equation}\label{Gmay33}
     p(\xi)|\phi(\xi,E_j)|^2\leq O(1)\xi^{-2a-\frac{2\alpha}{2+\alpha}}\leq O(1)\xi^{-1-\epsilon}.
\end{equation}
It implies  $\phi(\xi,E_j)\in L^2(\R ^+,p(\xi)d\xi)$ and then $u\in L^2(\R^+)$.   The Theorem follows immediately    by  applying  $a=\frac{2-\alpha}{2+\alpha}$  to \eqref{Gmay31}.

\end{proof}


 \section*{Acknowledgments}
   The research was supported by  NSF DMS-1700314 and NSF DMS-1401204.

\end{document}